%% file: circlebundles.tex
\documentclass{amsart}
\usepackage{amssymb}
\usepackage{amsmath}
\usepackage{amsthm}
\usepackage[arrow,matrix,curve]{xy}
\usepackage{graphicx}
\usepackage{hyperref}
\newcommand\wt{\widetilde}
\newcommand\wc{\wt c}
\newcommand\wB{\wt B}
\newcommand\wE{\wt E}
\newcommand\wO{\wt O}
\newcommand\wP{\wt P}
\newcommand\wR{\wt R}
\newcommand\wV{\wt V}
\newcommand\wX{\wt X}
\newcommand\wY{\wt Y}
\newcommand\pt{\text{\textup{pt}}}
\newcommand\wh{\widehat}
\newtheorem{theorem}{Theorem}[section]

\newtheorem{lemma}[theorem]{Lemma}
\newtheorem{proposition}[theorem]{Proposition}
\newtheorem{corollary}[theorem]{Corollary}
\newtheorem{axioms}[theorem]{Axioms}

\theoremstyle{remark}
\newtheorem{definition}[theorem]{Definition}
\newtheorem{remark}[theorem]{Remark}

\newcommand\cK{\mathcal{K}}

\newcommand\cO{\mathcal{O}}
\newcommand\cP{\mathcal{P}}

\newcommand\bp{\mathbf{p}}
\newcommand\bC{\mathbb{C}}
\newcommand\bP{\mathbb{P}}

\newcommand\bR{\mathbb{R}}
\newcommand\bT{\mathbb{T}}
\newcommand\bZ{\mathbb{Z}}
\newcommand\Diff{\operatorname{Diff}}
\newcommand\End{\operatorname{End}}
\newcommand\Br{\operatorname{Br}}
\newcommand\Ad{\operatorname{Ad}}
\newcommand\ad{\operatorname{ad}}
\newcommand\Cl{\operatorname{C\ell}}
\newcommand\Homeo{\operatorname{Homeo}}
\newcommand\Isom{\operatorname{Isom}}
\newcommand{\bF}{{\mathbb F}_2}
\newcommand{\co}{\colon\,}
\newcommand\Ca{$C^*$-algebra}

\newcommand{\lp}{\textup{(}}
\newcommand{\rp}{\textup{)}}

\begin{document}
\title[T-duality for circle bundles]
{T-duality for circle bundles via\\ 
noncommutative geometry}
\author{Varghese Mathai}
\address{Department of Pure Mathematics, University of Adelaide,
Adelaide, SA 5005, Australia}
\email{mathai.varghese@adelaide.edu.au}
\urladdr{http://www.maths.adelaide.edu.au/mathai.varghese/}
\author{Jonathan Rosenberg}
\thanks{VM was partially supported by the Australian Research
Council Grant DP110100072. JR was partially supported by US National
Science Foundation Grants DMS-0805003 and DMS-1206159, and also
thanks the Department of Pure Mathematics at The University of Adelaide
for its hospitality during a visit in March 2012. The authors thank
the referee for useful suggestions.}
\address{Department of Mathematics,
University of Maryland,
College Park, MD 20742, USA}
\email{jmr@math.umd.edu}
\urladdr{http://www.math.umd.edu/\raisebox{-.6ex}{\symbol{"7E}}jmr}
\begin{abstract}
Recently Baraglia showed how topological T-duality can be extended to
apply not only to principal circle bundles, but also to non-principal
circle bundles. We show that his results can also be recovered via two
other methods: the homotopy-theoretic approach of Bunke and Schick,
and the noncommutative geometry approach which we previously used for
principal torus bundles. This work has several interesting byproducts,
including a study of the $K$-theory of crossed products by $\wO(2) =
\Isom(\bR)$, the universal cover of $O(2)$, and some interesting facts
about equivariant $K$-theory for $\bZ/2$. In the final section of this
paper, some of these results are extended to the case of bundles with
singular fibers, or in other words, non-free $O(2)$-actions. 
\end{abstract}
\maketitle

\section{Introduction}
\label{sec:intro}
T-duality in physics (the ``T'' standing for ``target space'')
is a duality between two string theories, that
interchanges ``winding'' and ``momentum'' modes and replaces tori 
in the spacetime target space by
their duals. (\emph{Duality} here means a symmetry that relates
theories that mathematically look different but that predict the same
physics.) Topological T-duality, which began from \cite{MR1265457}
and \cite{MR2080959,MR2116165},
aims to capture the basic \emph{topological} content in T-duality.
It has expanded to a vast literature, but for a
general survey, one can see the book \cite{MR2560910}. Roughly
speaking, there are three main approaches to topological
T-duality: via homotopy theory (as in \cite{MR2130624}),
via noncommutative geometry (as in \cite{MR2116734}),
and via stacks, groupoids, and/or equivariant topology
(as in \cite{MR2491618,MR2797285,MW12}). In the usual
setup of topological T-duality, one is given a principal torus bundle
$E \xrightarrow{p} X$ together with a class $h\in H^3(E,\bZ)$, and
aims (maybe after imposing other conditions)
to produce from this data a ``dual bundle'' $E^\#
\xrightarrow{p^\#} X$ and a dual class $h^\#\in H^3(E,\bZ)$.  From the
point of view of the original situation in physics, $E$ here is
spacetime for a type II superstring theory, $h$ is the H-flux (the
cohomology class of the field strength of the
Neveu-Schwarz B-field), and we are suppressing non-topological data
such as the actual B-field and the metric.

Recently, David Baraglia \cite{MR2985336,Baraglia1} showed that this setup can
be generalized to the case where $E\xrightarrow{p} X$  is a
\emph{non-principal} circle bundle, and even \cite{Baraglia2} extended
many of his results to affine torus bundles of arbitrary rank.

Our aim in this paper is to give new approaches to Baraglia's results
on circle bundles, that
show how they also fit into the topological framework initiated by
Bunke and Schick in \cite{MR2130624} and the noncommutative geometry
approach which we used in \cite{MR2116734,MR2222224}.  This grew out
of a larger project, still in progress, of trying to identify a good
larger category of ``generalized bundles'' that incorporates the
bundles of \cite{MR2511635,MR2520349} and is well-adapted to the 
noncommutative geometry approach to topological T-duality.

We begin with a few definitions. 
\begin{definition}
\label{def:circlebundle}
A \emph{principal circle bundle} over
a space $X$ is a map $E\xrightarrow{p} X$, where $E$ is a space
equipped with a free action of $\bT$ (the unit circle in the complex
plane, with its group structure coming from complex multiplication)
for which $p$ is the projection onto the orbit space $E/\bT$ and the
map is locally trivial, in the sense that for any $x\in X$, there is
an open set $U$ containing $x$ for which one has a commuting diagram
\[
\xymatrix{p^{-1}(U) \ar[rr]^{\cong} \ar[dr]^p && U \times \bT
\ar[dl]_{\text{pr}_1}\\
& \,U. &}
\]
A \emph{non-principal circle bundle} over $X$, on the other hand, is a
topological fiber bundle over $X$ with fiber $S^1$ and structure group
$\Homeo(S^1)$. Since the inclusions $O(2)=\Isom(S^1)\hookrightarrow 
\Diff(S^1)\hookrightarrow \Homeo(S^1)$ are
homotopy equivalences, as observed in \cite{Baraglia1},
there is no loss of generality in assuming that
the bundle has structure group $O(2)$.  Thus the difference between
principal and non-principal circle bundles simply amounts to replacing
$SO(2)$ (which can be identified with $\bT=U(1)$ in the usual way) 
by $O(2)$. It is convenient to view $O(2)$ as $SO(2)\rtimes
\{1,j\}$, where $j$ is the reflection matrix $$\begin{pmatrix}
  1&0\\0&-1 \end{pmatrix}.$$ Thus for reasonable spaces $X$ (in
applications, $X$ will always be a manifold of finite homotopy type),
a non-principal circle bundle over $X$ is classified by a map
$\varphi\co X\to
BO(2)$ up to homotopy. This in turn gives rise to a principal
$O(2)$-bundle $\pi\co\wE\to X$ and thus a commuting diagram
\[
\xymatrix{\wE \ar[r]^{\wc} \ar[d]^{\wt p} \ar[rd]^\pi & E\ar[d]^p\\
\wX \ar[r]^{c} & \, X .}
\]
Here $E\xrightarrow{p} X$ is the original non-principal circle bundle,
$\wc$ and $c$ are two-fold covering maps, and $\wE\xrightarrow{\wt p}
\wX$ is a \emph{principal} circle bundle gotten by taking the orbit
space of $\wE$ under the subgroup $SO(2)$ of $O(2)$.  The covering
maps $\wc$ and $c$ come from dividing out by the action of $j$. The
covering $\wX \xrightarrow{c} X$ is classified by an element $\xi=w_1(p)
\in H^1(X,\bF)$ ($\bF$ denoting the field of two elements), or
equivalently, by the homotopy class of the 
composite $X\xrightarrow{\varphi} BO(2) \xrightarrow{Bq} B\bZ/2 \simeq
K(\bZ/2, 1)$, where $Bq$ is the map on classifying spaces induced by
the quotient map of topological groups $q\co O(2) \to \bZ/2$. 
Then the isomorphism class of the bundle $p$ is determined by the
combination of $\xi\in H^1(X;\bF)$ and $c_1(\wt p)\in H^2(\wX, \bZ)$,
which we can also view as living in $H^2(X, \bZ_\xi)$, where
$\bZ_\xi$ is the local coefficient system on $X$ locally isomorphic to
$\bZ$ and determined by $\xi$. We call this the Chern class $c_1(p)$.
\end{definition}
\begin{remark}
\label{rem:Gysin}
As with principal $S^1$ bundles, a non-principal circle bundle
$E\xrightarrow{p} X$ comes with a Gysin sequence, that arises from the
fact that the Leray-Serre spectral sequence for computing the cohomology
of the total space of the fibration $S^1\to
E \xrightarrow{p} X$ has only two rows and thus
collapses at $E_3$. However, as the fundamental
group of the base $X$ acts trivially on $H^0$ of the fiber but
nontrivially on $H^1(S^1)\cong \bZ$, the spectral sequence has
$E_2^{p,0} = H^p(X, \bZ)$ but $E_2^{p,1} = H^p(X, \bZ_\xi)$, the 
cohomology of $X$ with local coefficients. Recall that
$\bZ_{\xi}\otimes \bZ_{\xi} \cong \bZ$ as local coefficient systems,
so the cup product of two twisted cohomology classes lives back in
untwisted cohomology. The Gysin sequence thus
takes the form
\[
\cdots \to H^p(X, \bZ_\xi) \xrightarrow{\cup c_1(p)} H^{p+2}(X,
\bZ)  \xrightarrow{p^*} H^{p+2}(E, \bZ) \xrightarrow{p_!}  
H^{p+1}(X, \bZ_\xi) \xrightarrow{\cup c_1(p)} \cdots .
\]
(Note the alternation between twisted and untwisted cohomology.)
\end{remark}

With these preliminaries, we can now state one of Baraglia's main
results:
\begin{theorem}[{Baraglia \cite[Proposition 4.1]{Baraglia1}}] 
\label{thm:Baraglia1}
For any circle bundle 
$E\xrightarrow{p} X$ as above, with invariants $\xi=w_1(p)
\in H^1(X,\bF)$ and $c_1(p)\in H^2(X, \bZ_\xi)$, and for any choice
of ``H-flux'' $h\in H^3(E,\bZ)$,
there is a unique T-dual circle bundle $E^\#\xrightarrow{p^\#} X$ 
and there is a unique T-dual H-flux $h^\#\in H^3(E, \bZ)$, characterized by the
following axioms:
\begin{enumerate}
\item T-duality is natural, so the T-dual of a pull-back is the
  pull-back of the dual.
\item T-duality is involutive, so that $(E, p, h)$ is the T-dual of
  $(E^\#, p^\#, h^\#)$.
\item $\xi=w_1(p)=w_1(p^\#)$, so the double covering of $X$ defined by
  the T-dual is the same as $\wX \xrightarrow{c} X$ defined by $p$.
\item $p_!(h) = c_1(p^\#)$ and $(p^\#)_!(h^\#) = c_1(p)$.
\item $h$ and $h^\#$ agree after pull-back to the ``Poincar\'e
  bundle'' $E\times_X E^\# \to X$.
\end{enumerate}
\end{theorem}
In Section \ref{sec:homotopy} below, we will show how to give another
proof of Theorem \ref{thm:Baraglia1} following the outline of the
method in \cite{MR2130624}.  In Section \ref{sec:NCG}, we will give
still another proof using {\Ca} crossed products, which also leads to
the matching of twisted $K$-groups for $(p,h)$ and $(p^\#,h^\#)$ as
found by Baraglia in \cite[Proposition 6.1]{Baraglia1}.

Our method of proof uses an analogue of Connes' ``Thom Isomorphism
Theorem'' \cite{MR605351} for computing $K$-theory for the crossed
product of a {\Ca} by an action of $\wO(2) = \bR \rtimes \{1,
j\}$. The proof of this theorem, which uses a construction in
equivariant $K$-theory from \cite{MR3044609}, occupies Section
\ref{sec:KThy}. Finally, Section \ref{sec:fibr} deals with extension
of the main results to the case of $O(2)$-bundles with singular
fibers, that is $O(2)$-actions that are not necessarily free,
following a method from \cite{MW12}.

\section{A Homotopy-Theoretic Construction}
\label{sec:homotopy}

The following definition is basically copied from \cite[Definition
  2.1]{MR2130624}. 
\begin{definition}
\label{def:pair}
Let $X$ be a topological space having the homotopy type of a CW
complex. A \emph{pair} over $X$ will consist of a (non-principal) circle bundle
$E\xrightarrow{p} X$ in the sense of Definition
\ref{def:circlebundle}, together with a class $h\in H^3(E,\bZ)$.  Two
pairs $(E\xrightarrow{p} X, h)$ and $(E'\xrightarrow{p'} X, h')$ are said
to be \emph{isomorphic} if there is a bundle isomorphism $\varphi$ making the
diagram 
\[
\xymatrix{E \ar[rr]^{\varphi}_{\cong} \ar[dr]^p && E' 
\ar[dl]_{p'}\\
& X &}
\]
commute and satisfying $\varphi^*(h') = h$. Let $\cP(X)$ denote the
set of isomorphism classes of pairs over $X$.  It is obvious that the map
$X\mapsto \cP(X)$ is a homotopy-invariant functor on the
category of spaces with the homotopy type of a CW complex, with values
in pointed sets.  ($\cP(X)$ has no special structure other than that
of a set, but it does has a distinguished basepoint, namely the class
of the trivial pair $(X\times S^1 \xrightarrow{\text{pr}_1} X, 0)$.)
\end{definition}
The following result is then a slight modification of
\cite[Proposition 2.6]{MR2130624}. 
\begin{theorem}
\label{thm:classspace}
The functor $X\mapsto \cP(X)$ is a representable functor on the
homotopy category of spaces with the homotopy type of a CW
complex. More precisely, there is a CW complex $R$, unique up to
homotopy equivalence, with the following properties:
\begin{enumerate}
\item There is a canonical pair $(E \xrightarrow{p} R, h)$ over $R$.
\item For any space $X$ with the homotopy type of a CW complex, every
  pair over $X$ is pulled back from the canonical pair via some map
  $X\to R$.
\item Two pairs over $X$ are isomorphic if and only if their
  classifying maps $X\to R$ are homotopic {\lp}note that $X$ need not
  be pointed, so this is unbased homotopy{\rp}.
\end{enumerate}
Furthermore, we can make the homotopy type of $R$ precise. The space
$R$ is connected, and its universal cover $\wR$ is the classifying space in
\cite[Definition 2.4]{MR2130624}, namely a two-stage Postnikov system
\[
\xymatrix{K(\bZ,3) \ar[r] & \wR \ar[d]\\
& K(\bZ,2)\times K(\bZ,2)}
\]
with $k$-invariant $k\in H^4(K(\bZ,2)\times K(\bZ,2), \,\bZ)$ given by
the cup product of the canonical classes in the two $H^2(\bZ, 2)$
factors.

However, the space $R$ has $\pi_1(R)\cong \bZ/2$, and $R$ is not simple;
$\pi_1(R)$ acts non-trivially on $\pi_2(R)\cong \bZ^2$, with the
generator of $\pi_1(R)$ acting on $\pi_2(R)\cong \bZ^2$ by $-1$.
\end{theorem}
\begin{proof}
We begin by constructing the space $R$. Recall that the classifying
space of the orthogonal group $O(2)$ comes with a fibration $BSO(2)
\to BO(2) \to B(\bZ/2)$, and since the quotient group $\bZ/2$ acts
non-trivially on $SO(2)$, this fibration is nontrivial; in fact,  the
space $BO(2)$ is not simple. We can describe it as a Borel construction
$\bC\bP^\infty \times_{\bZ/2} E(\bZ/2)$, where $\bZ/2$ acts on $BSO(2)
= \bC\bP^\infty$ by complex conjugation (this changes the sign of the
generator of $\pi_2$) and freely on $E(\bZ/2)$ as usual.

Next consider a closely related space $B= (\bC\bP^\infty \times
\bC\bP^\infty) \times_{\bZ/2} E(\bZ/2)$, where $\bZ/2$ acts on 
both copies of $\bC\bP^\infty$ by complex conjugation. Note that there
is a fibration $K(\bZ,2) \to B\xrightarrow{\psi}
BO(2)$, where the double cover of $\psi$ is the identity on the second
$\bC\bP^\infty$ factor and the fiber is the first $\bC\bP^\infty$ factor.
The universal cover $\wB$ of
$B$ is homotopy equivalent to $K(\bZ^2, 2)$. Since $B$ is obtained via
the Borel construction, the cohomology ring of $B$ is
$H^*_{\bZ/2}(\bC\bP^\infty \times \bC\bP^\infty, \bZ)$, which can be
computed via the spectral sequence with $E_2^{p,q} = H^p(\bZ/2,
H^q(\bC\bP^\infty \times \bC\bP^\infty, \bZ)) = H^p(\bZ/2, \bZ[c_1,
  c_2])$, where the generators $c_1$ and $c_2$ are each in degree
$q=2$. The generator of $\bZ/2$ changes the signs of $c_1$ and $c_2$,
and is thus given by multiplication by $(-1)^k$ on monomials in $c_1$
and $c_2$ of total degree $q=2k$. Thus $E^{0,4}_2 = H^0(\bZ/2,
\bZ c_1^2 \oplus \bZ c_1c_2 \oplus \bZ c_2^2)\cong \bZ^3$, since the
$\bZ/2$ action  is trivial. From this one can see that the edge
homomorphism $H^4(B, \bZ) \to H^4(\bC\bP^\infty \times \bC\bP^\infty,
\bZ)$ is split surjective. Thus there is a principal $K(\bZ, 3)$
bundle over $B$ 
\[
\xymatrix{K(\bZ,3) \ar[r] & R \ar[d]^\eta\\
& B}
\]
with $k$-invariant corresponding to $c_1c_2\in
H^4(\bC\bP^\infty \times \bC\bP^\infty, \bZ)$, and this defines the
space $R$ up to homotopy equivalence.

Via the composition $R\xrightarrow{\eta} B \xrightarrow{\psi}
BO(2)$, we have a map $R\to BO(2)$ and
thus a principal $O(2)$ bundle over $R$ and hence a (non-principal)
circle bundle $E\xrightarrow{p} R$. Since $\pi_1(R)\cong \bZ/2$ and
the composite $R\xrightarrow{\psi\circ\eta} BO(2) \to B(\bZ/2)$
induces an isomorphism on $\pi_1$, the 
Stiefel-Whitney class of $p$ is the generator of $H^1(R,\bF)\cong \bF$
and in particular is non-trivial. We have a commuting diagram
\[
\xymatrix{\wE \ar[r]^{\wc} \ar[d]^{\wt p} & E\ar[d]^p\\
\wR \ar[r]^{c} & \, R ,}
\]
where it is clear that $\wR$, the double cover and also the universal
cover of $R$, is the classifying space of Bunke and Schick discussed in
the statement of the theorem. Here $\wE$ is the circle bundle over
$\wR$ identified by Bunke and Schick, so it has the homotopy type $\wE
\simeq K(\bZ,2)\times  K(\bZ,3)$.  Note that $\bZ/2\cong \pi_1(E)$ acts
trivially on $H^0$ and $H^3$, but nontrivially on $H^2$. Thus one can
see that $E\simeq BO(2)\times K(\bZ,3)$. Projection onto the second
factor $K(\bZ,3)$ thus defines a canonical class $h\in H^3(E,
\bZ)$. So we have a canonical pair $(E\xrightarrow{p}R, h)$ over $R$,
and for any space $X$ and map $f\co X\to R$, we get an induced diagram
\[
\xymatrix{f^* E \ar[r]^{\wt f} \ar[d]^{f^*p} & E\ar[d]^p\\
X \ar[r]^{f} & R }
\] 
and a pair $(f^*E, \wt f^*(h))$ over $X$, which up to isomorphism only
depends on the homotopy class of $f$.  We need to prove universality.

First of all, the universal cover $\wR$ is what it has to be.  For if
a space $X$ is simply connected, then any circle bundle over $X$ is
automatically principal, since $w_1$ must vanish, and a pair in the
sense of Definition 
\ref{def:pair} is the same as a pair in the sense of \cite[Definition
2.1]{MR2130624}. Thus since our $\wR$ is the classifying space of
Bunke-Schick, we have (again assuming $X$ simply connected) 
\[
\cP(X) =_{\text{Bunke-Schick}} [X, \wR] \cong [X, R],
\]
since any map $X\to R$ has an essentially unique lifting to
$\wR$. (Really there are two lifts, since a basepoint in the image of
$X$ has two inverse images in $\wR$, but the two lifts are freely homotopic.)

Now suppose we have a pair $(E_1 \xrightarrow{p_1} X, h_1)$ over
$X$. Then we have the Stiefel-Whitney class $\xi=w_1(p_1)$, as well as two
classes in $H^2(X, \bZ_\xi)$: $c_1(p_1)$ and $(p_1)_!(h_1)$.  These
three classes correspond precisely to a homotopy class of maps
$g\co X \to B$.  (Recall that $B$ is a $\bZ/2$-twisted version of
$K(\bZ,2) \times K(\bZ,2)$.) However, the Gysin sequence (Remark
\ref{rem:Gysin}) gives us two additional pieces of information. First
of all, the cup product of $c_1(p_1)$ and $(p_1)_!(h_1)$ vanishes in
$H^4(X, \bZ)$. That says that $g$ factors through the homotopy fiber
of a map $k\co B\to K(\bZ, 4)$ associated to the product of the two
twisted canonical classes. Secondly, the exact sequence says that if
$(p_1)_!(h_1)$  is known, then $h_1$ is determined modulo
$(p_1)^*(H^3(X, \bZ))$. These pieces of information say exactly that
$g$ has a lifting to a map $f\co X\to R$, and that the homotopy class
of $f$ completely determines the pair $(E_1 \xrightarrow{p_1} X, h_1)$.
\end{proof}
\begin{lemma}
\label{lem:Tdualityinvol}
The classifying space $R$ of Theorem \ref{thm:classspace} has a unique
homotopy involution interchanging the two $K(\bZ,2)$'s and preserving
the homotopy class of the map $w_1\co R\to K(\bZ/2, 1)$ given as the
composite 
\[
R \xrightarrow{\eta} B \xrightarrow{\psi} BO(2) \xrightarrow{Bq} B\bZ/2.
\]
\end{lemma}
\begin{proof}
This is immediate from the explicit description of $R$; the key point
is that the map $k\co B \to K(\bZ, 4)$ is symmetric under the
interchange of the $K(\bZ,2)$ factors.
\end{proof}
\begin{corollary}
Baraglia's Theorem \ref{thm:Baraglia1} holds; that is, every pair has
a unique T-dual pair satisfying the axioms of Theorem \ref{thm:Baraglia1}.
\end{corollary}
\begin{proof}
The T-dual is gotten by composing the classifying map into $R$ with
the involution of Lemma \ref{lem:Tdualityinvol}. Property (1),
naturality, and property (2), involutiveness,
are immediate.  Property (4) follows from the fact that the
involution switches the two copies of $K(\bZ,2)$, which from the proof of
Theorem \ref{thm:classspace} correspond to $c_1(p)$ and $p_!(h)$,
respectively. Property (3) follows from the fact that the involution
of Lemma \ref{lem:Tdualityinvol} preserves $w_1$.  It remains just to
check the last property, (5). This follows just from checking the
universal example $R$. If $E^\# \xrightarrow{p^\#} R$ is the circle
bundle over $R$ T-dual to the original one (under the involution of
Lemma \ref{lem:Tdualityinvol}), then it is easy to see that
$E\times_R E^\# $ is just $K(\bZ, 3) \times K(\bZ/2, 1)$, and the
canonical classes $h\in H^3(E,\bZ)$ and $h^\# \in H^3(E^\#,\bZ)$ both
pull back to the canonical class in $H^3(K(\bZ, 3), \bZ)$.  Since every
T-duality is pulled back from the universal one, the theorem follows.
\end{proof}

\section{An Approach via Noncommutative Geometry}
\label{sec:NCG}

Now we aim to reproduce topological T-duality for non-principal circle
bundles using the basic idea from \cite{MR2116734}, which depends on
{\Ca} crossed products.  For this section, we will assume $X$ is a
smooth manifold of finite homotopy type and $E\xrightarrow{p} X$ is a circle
bundle, not necessarily principal.  These conditions are much stronger
than they need to be\footnote{$X$ need not be a manifold and doesn't
  really have to have finite homotopy type.}, 
but they are satisfied in physics applications,
and they make the proof of Theorem \ref{thm:Glift} much easier.
Let $\wE\to X$ be the principal $O(2)$-bundle associated to $p$. As in
Definition 
\ref{def:circlebundle}, define the double coverings $c$ of $X$ and
$\wc$ of $E$ associated to $\xi = w_1(p)$. Given $h\in H^3(E,\bZ)$, pull it
back to $(\wc)^*(h)\in  H^3(\wE,\bZ)$, and form the associated
stable continuous-trace algebra $A = CT(\wE, (\wc)^*(h))$. We have a
free action of $O(2)$ on $\wE$, but it may not lift to an action of
$O(2)$ on $A$. We therefore consider the universal cover $G = \wO(2) =
\bR\rtimes H$ of $O(2)$, where $H\cong \bZ/2$ lifts the subgroup $\{1,
j\}$ of $O(2)$. Here, as before, $j$ is the reflection matrix $\begin{pmatrix}
  1&0\\0&-1 \end{pmatrix}$. Note that $G$ is the affine group of the
line. In this context, we have the following.
\begin{theorem}
\label{thm:Glift}
In this situation of a principal $O(2)$-bundle $\wE\to X$ and a class
in $H^3(\wE, \bZ)$ pulled back from the non-principal circle bundle
$E$ over $X$, the free action of $O(2)$ on $\wE$ lifts to an action of
$G = \wO(2)$ on the stable continuous trace algebra $A = CT(\wE, (\wt
c)^*(h))$.   Furthermore, the action of $G$ on $A$ is uniquely
determined up to exterior equivalence once $h\in H^3(E,\bZ)$ is fixed.
\end{theorem}
The proof uses the obstruction theory developed in \cite{MR1446378}, which
gives necessary and sufficient conditions for such a lifting.  The key
point is that we can compare the obstructions for $G$ with the similar
obstructions for a $\bZ/2$-action lifting the free action of $\bZ/2$
on $\wE$ with quotient space $E$.  Since the obstruction theory is
based on the group cohomology of {C.} {C.} Moore \cite{MR0414775},
further developed by {D.} Wigner \cite{MR0338132}, we start with three
lemmas about these cohomology groups.
\begin{lemma}
\label{lem:cohomfinite}
Let $H$ be a finite group and let $A$ be an $H$-module which is a
vector space over a field $k$ of characteristic zero with a linear
$H$-action.  Then $H^q(H, A)=0$ for $q>0$.
\end{lemma}
\begin{proof}
This is a standard fact in homological algebra.  The point is that the
group ring $kH$ is semisimple (by Maschke's Theorem), because of the
fact that $k$ has characteristic $0$ (in fact it would suffice for the
characteristic to be relatively prime to the order of $H$), so the functor
$A\mapsto A^H$ is exact, and has no higher derived functors.
\end{proof}
\begin{lemma}
\label{lem:cohomR}
Let $H=\bR$ be the real line viewed as a topological group. 
\begin{enumerate} 
\item Let $A$ be a countable discrete $\bR$-module. Then $H^q(H, A)=0$
for $q>0$ {\lp}and is just $A$ itself for $q=0${\rp}.
\item Let $A$ be an $\bR$-module which is a complete metrizable
topological vector space.
Then $H^q(H, A)$ {\lp}for Moore's cohomology theory{\rp} is a real vector
space which vanishes for $q>1$.
\end{enumerate}
\end{lemma}
\begin{proof}
For the first fact, see \cite[Corollary 4.3]{MR2116734},
which is based on \cite{MR0338132}. For the second, see
\cite[Corollary 4.7]{MR2116734}. 
\end{proof}
\begin{lemma}
\label{lem:cohomCT}
Let $G = \bR\rtimes H$, $H=\{1,j\}\cong \bZ/2$ acting on $\bR$ by
$j\cdot t = -t$, $t\in \bR$, and let $X$ be a locally compact second
countable $G$-space with finite homotopy type.  
Give $C(X, \bT)$ the $G$ action coming from the
action of $G$ on $X$.  Then the maps $H^q(G, C(X,\bT))\to H^q(H,
C(X,\bT))$ induced by the inclusion $H\hookrightarrow G$ are
isomorphisms for $q=2,\,3$.
\end{lemma}
\begin{proof}
Since $X$ has finite homotopy type, $M=H^s(X,\bZ)$ is discrete and
countable for any $s$, and $\bR\triangleleft
G$ acts trivially on it. Thus by Lemma \ref{lem:cohomR}(1),  
$H^q(\bR, M)=0$ for $q>0$.

Now we have an exact sequence of topological $G$-modules
\cite[\S4.2, equation (4)]{MR2116734}
\begin{equation}
\label{eq:CXT}
0 \to H^0(X, \bZ) \to C(X,\bR) \to C(X,\bT) \to H^1(X,\bZ) \to 0.
\end{equation}
We have just observed that the $\bR$-cohomology of the discrete
modules $H^s(X, \bZ)$, $s=0,\,1$, vanishes past degree $0$. 
And $H^q(\bR, C(X, \bR))$ is a topological vector space
for $q=0,\,1$, and vanishes for $q>1$, by Lemma \ref{lem:cohomR}(2).
Now use the spectral sequence $H^p(H, H^q(\bR,C(X,
\bR)) \Rightarrow H^{p+q}(G, C(X, \bR))$. The spectral sequence has
$E_2^{p,q}=0$ for $q>1$. Since $H^q(\bR,C(X,\bR))$ is a real
vector space for $q=0,\,1$ and $H$ is
finite cyclic, we see by Lemma \ref{lem:cohomfinite} that
$E_2^{p,q}=0$ for $p>0$.  Thus (since we've also seen that
$E_2^{p,q}=0$ for $q>1$) $H^k(G, C(X,\bR))=0$ for all
$k>1$. Similarly $H^k(H, C(X,\bR))=0$ for all $k>0$ by Lemma
\ref{lem:cohomfinite}.

Now split \eqref{eq:CXT} into two short exact sequences
\begin{equation}
\label{eq:CXT1}
0 \to H^0(X, \bZ) \to C(X,\bR) \to C(X,\bT)_0 \to 1
\end{equation}
and
\begin{equation}
\label{eq:CXT2}
1 \to C(X,\bT)_0 \to C(X,\bT) \to H^1(X, \bZ) \to 0.
\end{equation}
From \eqref{eq:CXT1}, we find that the $G$- or $H$-cohomology of
$C(X,\bT)_0$ past 
degree $1$ coincides with the $H$-cohomology of the discrete module
$H^0(X, \bZ)$ one degree higher. The conclusion of the lemma now
follows by applying the long exact cohomology sequence of
\eqref{eq:CXT2} along with the spectral sequences $H^p(H, H^q(\bR,
H^s(X, \bZ)) \Rightarrow H^{p+q}(G, H^s(X, \bZ))$, $s=0,\,1$.
\end{proof}
\begin{proof}[Proof of Theorem \ref{thm:Glift}]
We could use the obstruction theory developed in \cite{MR1446378}, which
gives necessary and sufficient conditions for such a lifting.  The key
point is that we can compare the obstructions for $G$ with the similar
obstructions for a $\bZ/2$-action lifting the free action of $\bZ/2$
on $\wE$ with quotient space $E$.  But since the algebra $A$ is the
pull-back of the algebra $CT(E, h)$ along the covering map $\wc$,
such a $\bZ/2$-action exists by \cite{RaeWilliams}.  (In fact, by
\cite[\S6.2]{MR1446378}, pull-back induces an isomorphism $\Br_H(\wE)
\cong \Br(E)$ for the equivariant Brauer group.) Thus the
obstructions vanish for $\bZ/2$, and so we will see that this implies
they also vanish for $G$.

However the simplest argument for existence uses bundle theory and
connections. 
By \cite[Lemma 7.1]{MR2560910}, lifting the $G$-action on $\wE$ to a
$G$-action on $A$ is equivalent to lifting the $G$-action to the
principal $PU$-bundle associated to $A$.  Let $P \to E$ be a
principal $PU$-bundle with Dixmier-Douady class $h$; we can assume it
is smooth. Choose a connection for $P$ (recall that this gives a
canonical way to lift vector fields from $E$ to $P$) and take the
pull-back connection (under the covering map $\wc$)
on the pull-back bundle $\wP \to \wE$. The $\bR$-action on $\wE$ with
quotient $\wX$ is defined by a smooth integrable vector field $V$,
which lifts via 
the connection to a horizontal vector field $\wV$ on $\wP$. Since the
connection is pulled back under $\wc$ from $P$, it is 
$H$-invariant. In other words, the pull-back action of $j\in H$ on
$\wP$ preserves horizontal vectors.  Since the action of $j$
conjugates $V$ to $-V$, it must also
conjugate $\wV$ to $-\wV$ (since horizontal lifts of tangent vectors
are unique).  Thus $\wV$ integrates to an
$H$-equivariant action of $\bR$ on $\wP$, that is, to a $G$-action of
the type we require.

The alternative argument for existence uses \cite[Theorems 4.9 and
5.1]{MR1446378}, together with Lemma \ref{lem:cohomCT}. We note that
the domain and range of the maps $d_2$, $d_2'$, and $d_3$ of those
theorems are naturally the same for both $H$ and $G$.  A diagram chase
then shows that the obstructions are the same for both groups, and
since we have a pull-back action of $H$ on $A$, there must also be a
$G$-action. 

Now we come to the issue of uniqueness. For this we use the exact
sequence of \cite[Theorem 4]{MR2278062}, which is also contained in
\cite[Theorem 5.1]{MR1446378}. In fact, we use this exact sequence
twice, once for $G$ and once for $H$, plus the result $\Br_H(\wE)
\cong \Br(E)$ from \cite[\S6.2]{MR1446378}. The exact sequence,
together with the fact that we have a forgetful map $\Br_G(\wE)\to
\Br_H(\wE)$ induced by the inclusion $H\hookrightarrow G$, gives a
commuting diagram
\begin{equation}
\label{eq:relBr}\hspace*{-5pt}
\xymatrix@C-1pc{
H^2(\wE,\bZ)^G \ar[r]^(.4){d_2''} \ar@{=}[d] & H^2(G, C(\wE,
  \bT)) \ar[r]^(.6){\xi_G} \ar[d] & \ker F_G\ar[d] \ar[r]^(.35){\eta_G} &
  H^1(G,H^2(\wE,\bZ))  \ar@{=}[d]\ar[r]^{d_2'} & 
H^3(G, C(\wE,  \bT))\ar[d] \\
H^2(\wE,\bZ)^H \ar[r]^(.4){d_2''} & H^2(H, C(\wE,\bT)) \ar[r]^(.6){\xi_H}  & \ker
F_H  \ar[r]^(.35){\eta_H} &  H^1(H,H^2(\wE,\bZ))\ar[r]^{d_2'} & \,
H^3(G, C(\wE,  \bT)).} 
\end{equation}
Here $F_G$ is the forgetful map $\Br_G(\wE)\to \Br(\wE)$, and
similarly $F_H$ is the forgetful map $\Br_H(\wE)\to \Br(\wE)$.
The vertical equal signs are canonical isomorphisms by Lemma
\ref{lem:cohomR}(1). 
By Lemma \ref{lem:cohomCT},
the first and last downward arrows in \eqref{eq:relBr} are
also isomorphisms. It 
follows by \eqref{eq:relBr} and the Five-Lemma that the forgetful map
$\ker F_G\to \ker F_H$ is an isomorphism. But we already know that
$\Br_H(\wE) \cong \Br(E)$ via pull-back under $\wc$. So this proves the
uniqueness statement.
\end{proof}
\begin{theorem}
\label{thm:NCmethod}
In the situation of Theorem \ref{thm:Glift}, let $A = CT(\wE, {\wt
c}^*(h))$ and let $\alpha$ be an action of $G = \wO(2)$ lifting the
free action of $\bZ/2$ on $\wE$ with quotient space $E$. Then the
crossed product $A\rtimes_{\alpha} G$ is of the form $CT(E^\#, h^\#)$
for $p^\#\co E^\#\to X$ some {\lp}non-principal{\rp} circle bundle over $X$
and $h^\#\in H^3(\wE^\#, \bZ)$. Furthermore, $w_1(p^\#) = w_1(p)$, and
the other conditions in Theorem \ref{thm:Baraglia1} are also satisfied.
\end{theorem}
\begin{proof}
Note that $A\rtimes_{\alpha} G\cong (A\rtimes_{\alpha} \bR)\rtimes
H$. Since $A$ is the stable continuous-trace algebra associated to an
H-flux on a \emph{principal} $S^1$-bundle over $\wX$, the theory
developed in \cite{MR2116734} applies, and $A\rtimes_{\alpha} \bR$ is
a stable continuous-trace algebra with spectrum the T-dual principal
$S^1$-bundle $\wE^\#$ over $\wX$ and Dixmier-Douady class the T-dual
H-flux (for T-duality of principal bundles over $\wX$).  Since the
action of $H$ on $A$ induces the free $\bZ/2$-action on $\wX$ with
quotient $X$, the action of $H$ on the spectrum $\wE^\#$ of
$A\rtimes_{\alpha} \bR$ must also be free. Furthermore,
$A\rtimes_{\alpha} G\cong (A\rtimes_{\alpha} \bR)\rtimes H$ will have
continuous trace \cite[Theorem 1.1]{MR920145}, and will be stable
since $A$ is stable, and thus will be of the form
$CT(E^\#, h^\#)$ for some $h^\#\in H^3(E^\#, \bZ)$.  By Takai duality
(for $H$), $A\rtimes_{\alpha} \bR \cong CT(E^\#, h^\#)\rtimes
\widehat H$ for the dual action of $\widehat H\cong H$.
Via \cite[Proposition 1.5]{RaeWilliams}, we conclude that
$A\rtimes_{\alpha} \bR$ is the pull-back of $CT(E^\#, h^\#)$ along a
principal $\widehat H$-bundle. So
we get a commuting diagram of bundles 
\[
\xymatrix{\wE^\# \ar[r]^{\wc^\#} \ar[d]^{\wt p^\#} \ar[rd]^{\pi^\#} 
& E^\#\ar[d]^{p^\#}\\
\wX \ar[r]^{c} & \, X ,}
\]
and the spectrum $E^\#$ of $A\rtimes_{\alpha} G$ is a non-principal
$S^1$-bundle over $X$ with the same $w_1$ as $E$ (since the associated
double cover is again $\wX \xrightarrow{c} X$). 
The Dixmier-Douady class of $A\rtimes_{\alpha} \bR $ is
$(\wc^\#)^*(h^\#)$. 

Now we can check the condition of Theorems
\ref{thm:Baraglia1}. Conditions (1) and (3) are now
obvious. Condition (2), that T-duality is involutive, follows from
Takai duality,
since when we T-dualize a second time, we get 
\[
((\wc^\#)^*(A \rtimes_\alpha
G)) \rtimes_{\alpha^\#} G \cong ((A\rtimes_\alpha
\bR)\rtimes_{\wh\alpha}\bR) \rtimes H \cong A\rtimes H \cong CT(E, h).
\]
Condition (4) follows from the corresponding condition for T-duality
of principal circles bundles over $\wX$, along with the relationships
between $c_1(\wt p)$ and $c_1(p)$, etc.

To check condition (5), note that we get from
\cite[Proposition 2.1 and Theorem 2.2]{MR920145} a commuting diagram
of principal circle bundles 
\begin{equation}
\label{eq:Pbundle}
\xymatrix{
& (A\rtimes_{\alpha} \bZ)\widehat{\phantom{x}} \ar[rd]^{\wt p^*(\wt p^\#)}
     \ar[ld]_{(\wt p^\#)^*\wt p}
  & \\
\wE^\#\ar[rd]^{\wt p^\#}& & \wE \ar[ld]_{\wt p}\\
& \wX & .}
\end{equation}
Furthermore $A\rtimes_\alpha \bZ\cong (A\rtimes_\alpha
\bR)\rtimes_{\wh\alpha} \bZ$ has Dixmier-Douady class equal to the
pull-backs of both $\wc^*(h)$ and $(\wc^\#)^*(h^\#)$. This diagram
fits into a larger diagram of circle bundles and covering maps
\begin{equation}
\label{eq:Pbundle1}
\xymatrix{
& (A\rtimes_{\alpha} \bZ)\widehat{\phantom{x}}
  \ar[rd]\ar[ld]\ar@{.>}[r]
  & E^\#\times_X E \ar@{-->}[rd]\ar@{-->}[ld]|\hole &\\
\wE^\#\ar[rd]^{\wt p^\#}\ar@{.>}[r]^{\wt c^\#}& E^\#\ar@{-->}[rd]|\hole^(.3){p^\#}&
\wE \ar[ld]_(.3){\wt   p}\ar@{.>}[r]^{\wt c}& E\ar@{-->}[ld]_p\\ 
& \wX \ar@{.>}[r]^c & X &,}
\end{equation}
where for simplicity we've left off the labels on the maps in the top
half of the diagram and we
have used dotted arrows for $2$-fold covering maps and
dashed arrows for non-principal circle bundles. Now over the space
$E^\#\times_X E$ we have the continuous-trace algebra
$(A\rtimes_\alpha \bZ) \rtimes H$. The algebra $A\rtimes_\alpha \bZ$
is pulled back from $A$ by \cite[Theorem 2.5]{RaeWilliams}, and $A$ is
pulled back from $CT(E, h)$, so by commutativity of
\eqref{eq:Pbundle1}, $A\rtimes_\alpha \bZ$ is pulled back from the
algebra $(A\rtimes_\alpha \bZ) \rtimes H$ over $E^\#\times_X E$, which
in turn is pulled back from $CT(E, h)$. Thus the Dixmier-Douady
invariant of $A\rtimes_\alpha \bZ$ is the pull-back of $h\in H^3(E,
\bZ)$. A similar argument using the dual action shows that it is also
equal to the  pull-back of $h^\#\in H^3(E^\#, \bZ)$. So this proves
property (5) of Theorem \ref{thm:Baraglia1} and concludes the proof.
\end{proof}

\section{\texorpdfstring{$K$-Theory of Certain Crossed Products}{K-Theory of Certain Crossed Products}}
\label{sec:KThy}

\begin{definition}
\label{def:ROtwisting}
We begin with a fact about equivariant $K$-theory for a compact group
$H$, namely that it is \emph{$RO(H)$-graded} (see for example
\cite[Chapters IX, X, XIII, and XIV]{MR1413302}). Given a compact
group $H$, an $H$-{\Ca} $A$, and an 
orthogonal representation of $H$ on a finite-dimensional real vector
space $V$, we can twist $H$-equivariant $K$-theory of $A$ by $V$, getting
$K_i^{H,V}(A) = K_i(A \otimes C_0(V))$, where $H$ acts on the second
factor via the linear representation and acts on the tensor product by
the tensor product action.  Note that if $V$ happens to be
a complex vector space and the action of $H$ is complex linear, then
equivariant Bott periodicity gives an isomorphism $K_*^{H,V} \cong
K_*^H$. (This is also true more generally if $V$ is even-dimensional
over $\bR$ and if the action of $H$ factors through $\text{Spin}^c(V)$.)
And if $H$ acts trivially on $V$, $K_*^{H,V} \cong K_{*+\dim
  V}^H$. But in general, the groups $K_*^{H,V}$ are not the same as
$K_*^H$, even modulo a grading shift. In the noncommutative world,
another approach to the groups $K_*^{H,V}$ is possible via graded
Clifford algebras, since
$C_0(V)$ is $KK^H$-equivalent to $\Cl(V)$, the complex Clifford
algebra of $V$ viewed as a \emph{graded} $H$-algebra \cite[Theorem
  20.3.2]{MR1656031}. But this requires introducing graded {\Ca}s,
which we'd prefer to avoid.

For our purposes we will need only the group
$H=\bZ/2$, which has two real characters, the trivial character $1$ and
the non-trivial character $t$, the sign representation $-$ (where the
generator of the group acts by $-1$ on $\bR$).  Thus we have twisted
equivariant $K$-groups $K_*^{\bZ/2, -}$, which are discussed in
greater detail in \cite{MR3044609}. These are modules over the
representation ring $R=R(H)\cong \bZ[t]/(t^2-1)$, which has two
special complementary prime ideals, $I=(t-1)$ and $J=(t+1)$.
The coefficient groups for $K_*^{\bZ/2, -}$
are computed in \cite{MR2545608}, for example. It turns
out that $K_*^{\bZ/2, -}(\bC) \cong \bZ$ (actually $R/J$ as an $R$-module)
for $*\equiv0\pmod{2}$ and $\cong 0$ for $*\equiv 1\pmod{2}$. Twisting
twice brings us back  
to conventional equivariant $K$-theory since a direct sum of two
copies of the sign character is a complex representation, where
equivariant Bott periodicity applies.
\end{definition}
\begin{remark}
\label{rem:KTfailure}
Let $H=\bZ/2$, $R=R(H)$, and let $A$ be an $H$-algebra.
It is important to note that the equivariant $K$-groups $K_*^H(A)$
(even as $R$-modules) \emph{do not determine the groups}
$K_*^{H,-}(A)$ in general. This is due to the failure of the
equivariant K\"unneth Theorem for $K_*^H$ localized at prime ideals
$P$ containing the augmentation ideal $I$. (See
\cite[p.\ 235]{MR911880}.) Here is a way to construct a very specific
counterexample. It is known that the Cuntz {\Ca} $\cO_2$ has vanishing
$K$-groups and is very ``nice'' (nuclear and in the ``bootstrap
class'' where the universal coefficient theorem holds for
$KK$). Furthermore, it is known (originally due to Blackadar
(unpublished), but see \cite[Lemma 4.7]{MR2053753} for a specific
construction) that there are actions of $H$ on $\cO_2$ for which
$\cO_2\rtimes H$ has nontrivial, but uniquely $2$-divisible,
$K$-groups, say $K_0(\cO_2\rtimes H) \cong \bZ[\frac12]$. Choose such
an action and let $A=\cO_2\rtimes H$, which is an $H$-algebra under
the \emph{dual action} to the original action on $\cO_2$. (Here we are
identifying $H$ with its dual in the obvious way.) By Takai
Duality and the Green-Julg Theorem, $K_*^H(A) \cong K_*(A\rtimes H)
\cong K_*(\cO_2) = 0$, while $K_*(A)\ne 0$. Incidentally, one could
also construct an abelian $H$-{\Ca} $A$ with $K_*(A)=0$ but
$K_*^H(A)\ne 0$ by
using deep results from algebraic topology, and then proceed similarly
with $A$ in place of $\cO_2$; see
\cite[Lemma 5.7]{MR2560910} and \cite{MR3044609}. 
\end{remark}
\begin{proposition}
\label{prop:KTfailure}
With $H=\bZ/2$ and $A$ the $H$-{\Ca} defined in Remark
\ref{rem:KTfailure}, $K_*^H(A)=0$ but $K_*^{H,-}(A)\ne 0$.\qed
\end{proposition}
\begin{proof}
We have $K_*^{H,-}(A)=K_*(A\otimes C_0(\bR^-))$, Since $\bR^-$
equivariantly contains $\{0\}$ (a fixed point) and
$\bR^-\smallsetminus \{0\}$ is equivariantly homeomorphic to $H\times \bR$
(with the $H$-action interchanging the two connected components), we
get a short exact sequence of $H$-algebras
\[
0\to A\otimes C(H) \otimes C_0(\bR) \to A\otimes C_0(\bR^-) \to A \to
0
\]
and thus a long exact sequence of equivariant $K$-groups
\[
\cdots \to K_{*+1}^H(A)  \to K_{*+1}(A) \to K_*^{H,-}(A) \to K_*^H(A) \to
\cdots.
\]
This exact sequence appears as \cite[Proposition 3.3]{MR3044609}.
Since $K_*^H(A)=0$ (in all degrees), we obtain $K_*^{H,-}(A)\cong
K_{*+1}(A) \ne 0$.
\end{proof}
\begin{corollary}[of the proof]
\label{cor:twistedKH}
Let $H=\bZ/2$ and let $A$ be an $H$-{\Ca}. Then there is an exact
sequence of $K$-groups
\[
\cdots \to K_{*+1}^H(A)  \to K_{*+1}(A) \to K_*^{H,-}(A) \to K_*^H(A) \to
\cdots.
\]
In particular, if both $K_*^H(A)=0$ and $K_*(A)=0$, then
$K_*^{H,-}(A)= 0$. \hfill\qedsymbol
\end{corollary}
In order to use the theory developed in Section \ref{sec:NCG}, we need
to be able to compute the $K$-theory of crossed products by
$\wO(2)$. The following result, modeled on Connes's ``Thom Isomorphism
Theorem'' of \cite{MR605351}, is precisely what is needed.
\begin{theorem}
\label{thm:ConnesO2}
Let $A$ be a {\Ca} equipped with an action of $G=\wO(2)$. Note that by
restriction to the copy of $H=\bZ/2$ in $G$ gotten by lifting $\{1, j\}
\subset O(2)$, $A$ is also a $\bZ/2$-algebra. Then the $K$-theory of
$A\rtimes G$ is naturally isomorphic to $K_*^{\bZ/2, -}(A)$ as defined
in Definition \ref{def:ROtwisting}.
\end{theorem}
\begin{proof}
Note that $A$ being a $G$-algebra means that $A$ is both a
$\bZ/2$-algebra and an $\bR$-algebra, and that the actions $\alpha$ of $\bR$
and $\gamma$ of $\bZ/2$ are related by the commutation relation 
$\gamma_j(\alpha_t(a)) = \alpha_{-t}(\gamma_j(a))$.  (Here $j$ denotes the
nontrivial element of $\bZ/2$.) 
This commutation rule comes from the multiplication rule
in $G$. Note also for future use that the action $\gamma$ of $j$ on $\bR^-$
(the superscript is to emphasize that the action is by 
reflection) also induces an action of $j$ on the Pontrjagin dual group
$\widehat \bR \cong \bR^*$ (the dual vector space), again given by
reflection. By the isomorphism $A\rtimes G \cong (A\rtimes_\alpha
\bR)\rtimes_\gamma \bZ/2$ and the Green-Julg Theorem \cite{MR625361},
$K_*(A\rtimes 
G)$ is naturally isomorphic to $K_*^{\bZ/2}(A\rtimes_\alpha \bR)$. We want to
show that this in turn is isomorphic to $K_*^{\bZ/2}(A\otimes
C_0(\bR^-))$, the equivariant $K$-theory of the crossed product for the
\emph{trivial action} of $\bR$ on $A$. Here the $\bR$ is really the
dual group $\widehat \bR$ of the original copy of $\bR$, but since
$\bZ/2$ acts on $\widehat \bR$ 
here by the sign representation, the result then follows.  

First we construct natural homomorphisms 
\[
\phi_\alpha^*\co K_*^H(A\otimes C_0(\bR^-))
\to K_*^H(A\rtimes_\alpha \bR)
\]
as follows. Let 
$CA = C_0([0,1), A)\cong C_0([0,1)) \otimes A$ be the \emph{cone}
on $A$. Define an action of $G$ on $CA$ by letting $H$ act by its
original action on $A$ and act trivially on $[0,1)$. The real line
$\bR$ is defined to act on $CA$ by $\beta_t(f)(s) =
\alpha_{st}(f(s))$. This is compatible with the action of $H$ because
of the commutation relation $\gamma_j(\alpha_t(a)) = \alpha_{-t}(\gamma_j
(a))$, so we have an action of $G$ whose restriction to $\bR$ is $\beta$. 
Since the action of $G$ commutes with the action of $C^b([0,1))$ by
central multipliers, we have a $G$-equivariant exact sequence
\[
0 \to C_0((0,1), A) \to CA \xrightarrow{\text{eval\ at\ }0} A \to 0,
\]
and thus a short exact sequence of $H$-algebras
\begin{equation}
0 \to C_0((0,1), A) \rtimes_\beta \bR \to CA \rtimes_\beta \bR
\xrightarrow{\text{eval\ at\ }0} A \rtimes \bR \to 0. 
\label{eq:crprodext}
\end{equation}
But the action of $\bR$ on the copy of $A$ over $0\in [0,1)$ is
trivial, so the crossed product on the right is $A\otimes C_0(\bR^-)$
(as an $H$-algebra). Over each point of $(0,1)$, the action of $\bR$
on the associated copy of $A$ is just a rescaling of $\alpha$, and in
fact we claim that $C_0((0,1), A) \rtimes_\beta \bR \cong C_0((0,1)) \otimes
(A\rtimes_\alpha \bR)$ ($H$-equivariantly). Granting this, the
boundary map in the long 
exact $K^H$-theory sequence attached to \eqref{eq:crprodext} gives the
desired natural map $\Phi$.

To prove the claim, observe that both of the algebras $C_0((0,1), A)
\rtimes_\beta \bR$ and $C_0((0,1)) \otimes (A\rtimes_\alpha \bR)$ are
completions of the algebra $C_c((0,1)\times \bR, A)$, but for
different convolution products. If $s$ denotes the coordinate on
$(0,1)$ and $t$ the coordinate on $\bR$, the product is
\[
(f\star_1 g)(s, t) = \int f(s, u)\alpha_{su} g(s, t - u)\,du
\]
in the first case and
\[
(f\star_2 g)(s, t) = \int f(s, u)\alpha_u g(s, t - u)\,du
\]
in the second case. These products are intertwined by the map
$f\mapsto f'$, where $f'(s, t) = s^{-1}f(s, s^{-1}t)$, which is
clearly $H$-equivariant, and one can
easily see this extends to the desired isomorphism of crossed products.

Next, it is clear that the maps $\phi^*_\alpha$ satisfy the following 
(taken almost \emph{verbatim} from \cite[\S II]{MR605351}):
\begin{axioms}[following Connes]
\leavevmode
\begin{enumerate}
\label{ThomIsoAxioms}
\item If $A=\End(V)$ for some finite-dimensional complex
  representation $V$ of $H=\bZ/2$, equipped with
  the trivial action of $\bR$, then $\phi_\alpha^*$
  is the canonical isomorphism.
\item {\lp}Naturality{\rp} $\phi$ is natural for maps of $G$-algebras.
\item {\lp}Compatibility with suspension{\rp} If $S\alpha$ is the suspension
  of $\alpha$, acting on $SA= A \otimes C_0(\bR)$ (with trivial
  $G$-action on the second factor), and $s^j_A\co\! K_j(A)$ $\to
  K_{j+1}(SA)$ is the usual isomorphism, then
  $\phi^{j+1}_{S\alpha}\circ s^j_A = s^{j}_{A\rtimes_\alpha \bR}
  \circ \phi^{j}_{\alpha}$. 
\end{enumerate}
\end{axioms}

We follow the outline of the proof in \cite{MR605351},
starting with the following lemma.
\renewcommand{\qedsymbol}{}
\end{proof}
\renewcommand{\qedsymbol}{\textsquare}
\begin{lemma}[{Cf.\ \cite[II, Proposition 4]{MR605351}}]
\label{lem:fixedproj}
In the situation of Theorem \ref{thm:ConnesO2}, suppose $p$ is an $H$-fixed
self-adjoint projection in $A$ and $t\mapsto \alpha_t(p)$ is
smooth. Then there is an exterior equivalent action of $G$ on $A$
fixing $p$. 
\end{lemma}
\begin{proof}
Let $\delta$ be the unbounded derivation which is the infinitesimal
generator of the action of $\bR$. Then $\delta$ maps $A^\infty$, the
smooth vectors for the action of $G$, into itself, $\delta$ commutes
with $*$, and $\alpha_t(a) = e^{t \delta(a)} $. In the terminology of
\cite{MR605351}, $\delta(a) = i [M, a]$, where $M$ is a self-adjoint
``unbounded 
multiplier'' of $A$ (and a genuine multiplier of $A^\infty$), so
$\alpha_t(a) = e^{it\ad M}(a) = \Ad(e^{itM})(a)=e^{itM}ae^{-itM}$. (We've
changed Connes' $H$ to an $M$ to avoid confusion with the group $H=\bZ/2$.)
Since $p$ is $H$-fixed, $j\cdot p=p$. On the other hand, $j\cdot M = -
M$ because of the structure of $G$. We simply follow the same method as
in \cite{MR605351}. Let $M' = pMp + (1-p)M(1-p)$. This is again an
unbounded multiplier of $A$ (and a genuine multiplier of $A^\infty$),
and it clearly commutes with $p$. Furthermore, $M'$ is a
\emph{bounded} perturbation of $M$, and so defines an
exterior-equivalent action of $\bR$ fixing $p$ (as shown in
\cite{MR605351}). This action extends to 
an action of $G$ since $j\cdot M' = -M'$ (since $j\cdot p=p$ and $j\cdot M = -
M$).
\end{proof}
\begin{proof}[Proof of Theorem \ref{thm:ConnesO2} (cont'd)]
Now we return to Connes' strategy for proving Theorem
\ref{thm:ConnesO2}. Propositions 1, 2, and 3 in \cite[II]{MR605351} go
through without change and show that we may  
assume $A$ is unital, and that we are free to replace $A$ by a matrix
algebra over $A$ whenever necessary. Another simple observation is
that replacing $A$ by $A\otimes C_0(\bR^-)$ (with $\bR$ acting
trivially on the second factor) gives us maps
\[
\begin{aligned}
\phi_{\alpha\otimes 1}^*&\co K_*^H(A\otimes C_0(\bR^-)\otimes C_0(\bR^-))
\to K_*^H((A\otimes C_0(\bR^-))\rtimes_{\alpha\otimes 1} \bR)\\
\text{or} \quad
\phi_{\alpha\otimes 1}^*&\co K_*^H(A) \to  K_*^{H,-}(A\rtimes_\alpha
\bR),
\end{aligned}
\]
where we've used equivariant Bott periodicity once. Furthermore,
$\phi_\alpha^*$ is an isomorphism for all $A$ if and only if
$\phi_{\alpha\otimes 1}^*$  is an isomorphism for all $A$ (since we
can tensor with $C_0(\bR^-)$ to go back and forth from one to the
other).

To finish the proof, we apply Takai Duality, which gives an
$H$-equivariant isomorphism $(A\rtimes_\alpha \bR)\rtimes_{\widehat\alpha}
\widehat{\bR} \cong A \otimes \cK$, $\cK$ the algebra of 
compact operators on $L^2(\bR)$. Thus putting together $\phi^*$ for $A$ and
for $A\rtimes_\alpha \bR$, we get maps
\[
K_*^H(A) \xrightarrow{\phi_{\alpha\otimes 1}^*}
K_*^{H,-}(A\rtimes_\alpha \bR)
\xrightarrow{\phi_{\widehat\alpha}^*} K_*^H(A),
\]
and we just need to show this composite is the identity. Because of
compatibility with suspensions (property \ref{ThomIsoAxioms}(3)), it
is enough to do this for $K_0^H$. Without loss of generality, we may
assume $A$ is unital and just consider a class in $K_0^H(A)$
represented by a finitely generated projective $A$ module with
compatible $H$-action. By \cite[\S11.3]{MR1656031}, such a module can
be represented by an equivariant Murray-von Neumann equivalence class of
$H$-invariant projections $p\in \End(V)\otimes A$, for some
finite-dimensional $H$-module $V$. Without loss of generality, we may
replace $A$ by $\End(V)\otimes A$, with trivial $\bR$ action on the
first factor.  Apply Lemma \ref{lem:fixedproj}. This enables us to
change the $\bR$-action within the same exterior equivalence class so
that the action fixes $p$.  After these reductions, $p$ defines a
$G$-equivariant map $\bC\xrightarrow{p} A$.  Because of Axioms 
\ref{ThomIsoAxioms} (especially naturality), this reduces us to
checking the theorem for $\bC$ with trivial action, for which the
result is obvious by Axiom \ref{ThomIsoAxioms}(1).
\end{proof}
Now we apply this result to the situation of Section \ref{sec:NCG}.
\begin{corollary}[{cf.\ \cite[Proposition 6.1]{Baraglia1}}]
\label{cor:K-thy}
Let $X$ be a smooth manifold of finite homotopy type, and let $p\co
E\to X$ be a non-principal circle bundle over $X$ with associated
principal $O(2)$-bundle $\wE\to X$. Choose $h\in H^3(E, \bZ)$ and let
$A = CT(\wE, (\wc)^*(h))$ with the $G$-action $\alpha$ lifting the
free action of $O(2)$ on $\wE$ as in Theorem
\ref{thm:Glift}. Then there is a natural isomorphism $K^*(E, h)\cong
K^{*+1}(A \rtimes G)$. In other words, in the notation of Theorem
\ref{thm:NCmethod}, $K^{*+1}(E^\#, h^\#)\cong K^*(E, h)$.
\end{corollary}
\begin{proof}
By Theorem \ref{thm:ConnesO2}, we may replace $K^{*+1}(A \rtimes G)$
by $K^{*+1}_{H,-}(\wE, (\wc)^*(h))$. Thus we compute
\[
\begin{aligned}
K^{*+1}_{H,-}(\wE,\wc^*h) &= K^{*+1}_H(\wE \times \bR^-,\wc^*h) \\
&\text{(since the $H$-action is free and is a pull-back action)}\\
&\cong K^{*+1} (E \times \bR, h)\cong K^*(E, h). \qquad\mbox{\qedhere}
\end{aligned}
\]
\end{proof}

\input{circlefibration}

\bibliographystyle{hplain}
\bibliography{circlebundles}
\end{document}

%% file: circlefibration.tex
\section{Circle fibrations}
\label{sec:fibr}

It turns out to be necessary for some purposes
to consider compactifications of spacetime that are not necessarily
circle bundles,  but which also have singular fibers.
This happens for instance in the the Strominger-Yau-Zaslow (SYZ)
formulation \cite{MR1429831} of mirror symmetry. 
In the paper \cite{MW12}, the authors considered a simple case when
singular fibers exist, namely in   
the case of a smooth action of a circle on spacetime which is not
necessarily free. In this section, 
we consider a more general case of a smooth action, which is 
generically free but not
necessarily free, of the affine circle group $O(2)$ on a branched
double cover $\wE$ of spacetime $E$, with the reflection $j\in O(2)$
acting so that $E=\wE/\{1,j\}$. The $O(2)$-action gives 
rise to the (non-principal) circle ``fibration'' $E=\wE/\{1,j\}
\to X= \wE/O(2)$, which in general has singular fibers. 
We will extend the T-duality picture and isomorphism of the previous sections
to include this case. It will be a generalization, both of Baraglia's
Theorem \ref{thm:Baraglia1} and the main result in \cite{MW12}. 

Consider the diagonal action of $O(2)$ on $\wE \times EO(2)$, with
quotient the Borel construction $\wE_{O(2)} = \wE \times_{O(2)}
EO(2)$. Then $p\co\wE_{\bZ/2} =\wE \times_{\bZ/2} EO(2) \to 
\wE_{O(2)}$ is an honest (non-principal) $S^1$-bundle as discussed 
in Section \ref{sec:intro}\,\footnote{Strictly speaking, we've given up
  local compactness and homotopy finiteness in doing this, but this is
  not a serious issue. For most purposes one can replace $EO(2)$ by a
  finite-dimensional approximation which is actually a manifold.}, and
is classified by the pair of 
invariants, $\xi=w_1(p) \in H^1(\wE_{O(2)}, \bZ/2) 
= H_{O(2)}^1(\wE, \bZ/2)$ and $c_1(p) \in H^2(\wE_{O(2)}, \bZ_\xi)=
H_{O(2)}^2(\wE, \bZ_\xi)$. We also assume  
that $E$ comes with an $H$-flux $h$, which we pull back to an integral
$3$-cohomology class on
$\wE \times_{\bZ/2} EO(2)\simeq \wE \times_{\bZ/2} E\bZ/2 = \wE_{\bZ/2}$.
Then by  
the Gysin sequence (Remark \ref{rem:Gysin} in Section \ref{sec:intro}),
$p_!(h) \in H^2(\wE_{O(2)}, \bZ_\xi)= H_{O(2)}^2(\wE, \bZ_\xi)$. 

In analogy with the equivariant classification of principal circle bundles
\cite[Theorem C.47]{MR1929136},  we can deduce the following.
\medskip

\begin{theorem}[$O(2)$-equivariant circle bundles]
\label{thm:equivS1}
Let $M$ be a connected manifold equipped with a smooth action $\alpha$ of
$O(2)$. Define an $O(2)$-equivariant principal circle bundle over $M$
to mean a principal circle bundle $\pi\co Y \to M$, defined by a free
{\lp}smooth{\rp} action
$\alpha_2$ of $SO(2)$ on $Y$ with quotient space $M$, together with an
action $\alpha_1$ 
of $O(2)$ on $Y$, so that the restriction of $\alpha_1$ to
$SO(2)$ commutes with the action $\alpha_2$ defined by the bundle
structure, $\alpha_1$ lifts the original action $\alpha$ of $O(2)$ on $M$,
and the action under $\alpha_1$ of $j$ {\lp}as before, the
standard reflection matrix in $O(2)${\rp} satisfies
$\alpha_1(j)\alpha_2(t)\alpha_1(j)=\alpha_2(\bar t)$
{\lp}note the complex conjugation!{\rp}. This last condition means
that $\alpha_1(j)$ and $\alpha_2$ combine to give another
$O(2)$-action {\lp}which we again call $\alpha_2${\rp} on $Y$.
Then the pair of invariants, the equivariant first Stiefel-Whitney
class $\xi \in  H_{O(2)}^1(M, \bZ/2)$ defined by $\alpha$ as above and the  
equivariant first Chern class $c_1^{O(2)} \in H_{O(2)}^2(M, \bZ_\xi)$,
classify $O(2)$-equivariant principal circle bundles $(\pi\co Y\to M,
\alpha_1)$ over $M$ up to equivalence. 
\end{theorem}
\begin{proof}
By \cite[Theorem C.47]{MR1929136}, the $SO(2)$-equivariant principal
circle bundles over $M$ (what one gets by forgetting the action of $j$
everywhere) are classified by $c_1^{SO(2)} \in H_{SO(2)}^2(M,
\bZ)$. The problem is to augment this result to take the action of
$j$ into account. First suppose that an $O(2)$-equivariant principal
circle bundle $(\pi\co Y\to M, \alpha_1)$ over $M$ is given. We have
already defined $\xi$ as the cohomology class in $H_{O(2)}^1(M, \bZ/2)
= H^1(M_{O(2)}, \bZ/2)$ defining the double covering $M_{SO(2)}\to
M_{O(2)}$. Just thinking of $\pi\co Y\to M$ as an $SO(2)$-equivariant
circle bundle, we have the class $c_1^{SO(2)} \in H_{SO(2)}^2(M,\bZ) =
H^2(M_{SO(2)}, \bZ)$. Since $M_{SO(2)}$ is the double cover of
$M_{O(2)}$ defined by $\xi$ and the action $\alpha_1$ of $j$ does
\emph{not} commute with the bundle action $\alpha_2$ of $SO(2)$,but
rather \emph{anticommutes}, we can also think of this class as living
in $H^2(M_{O(2)}, \bZ_\xi)= H_{O(2)}^2(M, \bZ_\xi)$.

Now let's go the other way. Assume that $\xi$ (already defined by
$\alpha$) and $c_1^{O(2)} \in H_{O(2)}^2(M,\bZ_\xi)$ are given. 
Then we have the principal circle bundle $p\co M\times EO(2)\simeq M
\to M_{SO(2)}$ defined by $\alpha$ as well as another
principal circle bundle $\pi_1\co Y_1\to M_{SO(2)}$
defined by $c_1^{O(2)}$, or if you prefer, a principal $O(2)$-bundle
$\pi'_1\co Y_1\to M_{O(2)}$ with invariants $(\xi,
c_1^{O(2)})$. Forming the fiber product, we get a commuting 
diagram of principal circle bundles 
\begin{equation*}
\xymatrix{
& Y' \ar[dr]^{\pi_1^*p} \ar[dl]_{p^*\pi_1} & \\ M\times EO(2) \ar[dr]_{p} & & Y_1
  \ar[dl]^{\pi_1}\\ & M_{SO(2)} &.}
\end{equation*}
By \cite[Theorem C.47]{MR1929136}, we also have an
$SO(2)$-equivariant principal circle bundle $\pi\co Y\to M$ defined by
$c_1^{O(2)}$, and we can complete the previous diagram to a
commutative diagram 
\begin{equation*}
\xymatrix{
& Y \ar[dl]_{\pi} & \ar[l]_{\bp} Y' \ar[dr]^{\pi_1^*p} \ar[dl]_{p^*\pi_1} & \\ 
M \ar[dr] & \ar[l]^(.6){\text{proj}_1} M\times EO(2) \ar[dr]_{p} & & Y_1
  \ar[dl]^{\pi_1}\\ & M/SO(2) & \ar[l] M_{SO(2)} &.}
\end{equation*}

Here the top horizontal arrow $\bp$ (pointing to the left) is a map of
principal circle bundles, and is equivariant for the $SO(2)$-actions.
Since $\pi'_1\co Y_1\to M_{O(2)}$ is a principal $O(2)$-bundle, we also
have an action of $j$ on $Y'$ with the right intertwining properties,
and we just need to show that it descends to $Y$. We want to define
$j\cdot y = \bp(j\cdot y')$ when $\bp(y') = y$, so one just needs to
check that this is well-defined, i.e., doesn't depend on the choice of
$y'\in \bp^{-1}(y)$. But if also $\bp(y'_1) = y$, then $y'_1 = t\cdot
y'$ for some $t\in SO(2)$, and then 
\[
\bp(j\cdot y'_1) = \bp(j\cdot t\cdot y') =
\bp(\bar t \cdot j \cdot y')= \bp(j \cdot y'),
\]
as required.
\end{proof}

Using this result, we see that 
the pair of invariants, $\xi \in H_{O(2)}^1(\wE, \bZ/2)$ and $p_!(h) \in
H_{O(2)}^2(\wE, \bZ_\xi)$ 
determine an $O(2)$-equivariant principal
circle bundle $\wY \xrightarrow{\wt \pi} \wE$. Since the actions of $O(2)$
and $S^1$ on this space commute, we can divide out by the action of
$\{1,j\}\cong \bZ/2 \subset O(2)$, getting a principal $S^1$-bundle
$Y \xrightarrow{\pi} E$. Consider the commutative diagram (the
analogue of \eqref{eq:Pbundle})
\begin{equation*}
\xymatrix{
& Y \ar[dr]_{\pi} \ar[dl] & \\ E^\# \ar[dr] & & E \ar[dl]\\
& X &}
\end{equation*}
involving the singular spaces $X=\wE/O(2)$ and $E^\#= \wY/O(2)$.
We can replace the singular spaces by their Borel constructions by lifting the
above diagram to the following commutative diagram of actual circle bundles:
\begin{equation*}
\xymatrix{
& Y\times EO(2) \ar[dr]_{\pi\times 1} \ar[dl] & \\
\wY_{O(2)} \ar[dr]^{p^\#} & & \wE_{\bZ/2} \ar[dl]_p\\
& \wE_{O(2)} & .}
\end{equation*}
Now we can apply Baraglia's theorem Theorem \ref{thm:Baraglia1}
to deduce the existence 
of the T-dual $H$-flux $\widehat h \in H^3(Y_{O(2)} , \bZ) = 
H_{O(2)}^3(Y, \bZ)$ satisfying the properties listed in the theorem.
Summarizing all that we've done, we get the following:

\begin{theorem}
Let $\wE$ be a connected manifold with an action of $O(2)$ as above, 
with invariants $\xi=w_1(p)
\in H_{O(2)}^1(\wE,\bF)$ and $c_1(p)\in H_{O(2)}^2(\wE, \bZ_\xi)$. Then
for any choice of H-flux $h\in H^3_{\bZ/2}(\wE,\bZ)$,
there is a unique $O(2)$-equivariant circle bundle
$\wY\xrightarrow{\wt\pi} \wE$   
and a unique T-dual H-flux $h^\#\in H_{O(2)}^3(\wY,
\bZ)$, characterized by the following axioms:
\begin{enumerate}
\item T-duality is natural, so the T-dual of a pull-back is the
  pull-back of the dual.
\item $\xi=w_1(p)=w_1(p^\#)$, so the double covering of $\wE_{O(2)}$
  defined by the T-dual is the same as the original $\wE_{SO(2)} \to
  \wE_{O(2)} $. 
\item $p_!(h) = c_1(\widehat p)$ and $(\widehat p)_!(h^\#) =
  c_1(p)$ in $H_{O(2)}^2(\wE, \bZ_\xi)$. 
\item $h$ and $h^\#$ agree after pull-back to the correspondence
  space $Y\times EO(2) \simeq Y$. 
\end{enumerate}
\end{theorem}

Therefore we conclude, assuming for simplicity that $E$ is compact and
using Corollary \ref{cor:K-thy} and the appendix in \cite{MW12}, that 
$$ 
K^{\bullet}(E, h) \cong RK^{\bullet + 1}(\wY_{O(2)}, h^\#) \cong
K^{\bullet + 1}_{O(2)}(\wY, h^\#)^{\widehat{}}, 
$$
where $RK^{\bullet }(\wY_{O(2)}, h^\#)$ denotes the representable
twisted K-theory of the compactly generated space $\wY_{O(2)}$ and 
$K^{\bullet + 1}_{O(2)}(\wY, h^\#)^{\widehat{}}$ denotes the
$I(O(2))$-adic completion of  
the twisted equivariant K-theory $K^{\bullet + 1}_{O(2)}(\wY,
h^\#)$. Here $I(O(2))$ denotes the  
augmentation ideal of the representation ring $K^0_{O(2)}(\pt) =
R(O(2))$.